%% file: main_lncs.tex
\documentclass[letterpaper]{llncs}
\pdfoutput=1

\usepackage[linesnumbered,ruled]{algorithm2e}

\usepackage{amsmath}
\usepackage{amssymb}
\usepackage{amsfonts}
\usepackage{dsfont}

\usepackage{graphicx}

\usepackage{url}

\newenvironment{proofsketch}{\par{\noindent \emph{Proof Sketch.}}}{\qed\par}

\newcommand{\squishlist}{
   \begin{list}{$\bullet$}
    { \setlength{\itemsep}{0pt}      \setlength{\parsep}{1pt}
      \setlength{\topsep}{3pt}       \setlength{\partopsep}{0pt}
      \setlength{\leftmargin}{1.5em} \setlength{\labelwidth}{1em}
      \setlength{\labelsep}{0.5em} } }

\newcommand{\squishlisttwo}{
   \begin{list}{$\bullet$}
    { \setlength{\itemsep}{0pt}    \setlength{\parsep}{0pt}
      \setlength{\topsep}{0pt}     \setlength{\partopsep}{0pt}
      \setlength{\leftmargin}{2em} \setlength{\labelwidth}{1.5em}
      \setlength{\labelsep}{0.5em} } }

\newcommand{\squishend}{
    \end{list}  }

\newcommand{\Q}[2]{Q_{#1}^{({#2})}}
\newcommand{\Rnn}{\mathbb{R}_{\geq 0}}
\newcommand{\x}{\mathbf{x}}

\begin{document}

\mainmatter  

\title{Sequential Resource Allocation with Positional Costs}


\urldef{\mailsa}\path|bojhuang@microsoft.com|
\author{Bojun Huang}
\authorrunning{Bojun Huang}

\institute{Microsoft Research, Beijing, China \\
\mailsa
}
%
%

\onecolumn \maketitle

\begin{abstract}
We consider the problem of minimizing the total cost to run a sequence of $n$ tasks in the given order by $k$ agents under the \emph{positional cost model}. The cost to run a task not only depends on the intrinsic cost of the task itself, but also monotonically related to the position this task is in the working list of the agent assigned. Such a positional effect can naturally arise from the classic sum-of-completion-time minimization problems, and is also well motivated by the varying \emph{efficiency} when an agent works in reality (such as due to the learning effects or deteriorating effects). Also, it can be seen as a deterministic variant of the classic Baysian sequential decision making problems. This paper presents a simple and practical algorithm that runs in $O(k^2 n)$ time and minimizes the total cost of any problem instance consisting of two task types. The algorithm works by making greedy decision for each task \emph{sequentially} based on some stopping thresholds in a ``greedy-like" allocation simulation -- a working style coinciding with Gittins' optimal-stopping based algorithm for the classic Baysian multi-armed bandit problem.
\end{abstract}


\input{sec_introduction}
\input{sec_related}

\input{sec_preliminaries}

\input{sec_algorithm}

\input{sec_optimality}

\bibliographystyle{plain}
\bibliography{positional}

\newpage
\begin{appendix}
\input{sec_appendix}
\end{appendix}

\end{document}

%% file: sec_introduction.tex
\section{Introduction}

Consider the problem of minimizing the \emph{sum of completion time} to serve $n$ customers with $k$ identical service providers, where the customers are ordered by a \emph{first-come-first-serve} queue so that any service provider must serve the customers assigned to it in the same order as the queue (but two customers may be served simultaneously if they are assigned to different providers). Suppose a customer $i$ is the first customer of a provider that serves $m$ customers in total, since all the other $m-1$ customers of the same provider must wait for customer $i$ to complete in time $T_i$, the completion time ``caused by" customer $i$ is thus $m \cdot T_i$. Summing up over all customers, we have
\begin{equation} \label{eq:positional}
C_{total} = \sum_{i=1}^n f(p_i) \cdot T_i 
,\end{equation}
where $p_i$ is the \emph{position} of customer $i$ in the working list of its service provider, and $f(p_i)$ equals to the \emph{reversed position} of customer $i$, i.e., if a provider serves $m$ customers in total, then $f(p_i)=m+1-p_i$, which equals to $m$ for the first customer and $1$ for the last customer, in particular.

The sum-of-completion-time problem presented above is a special case of the general problem of \emph{Positional Allocation} studied in this paper. In general, we want to minimize the total cost to run a sequence of tasks by multiple agents under the \emph{positional cost model} as Eq. \ref{eq:positional} shows, where the cost to run a task not only depends on the independent cost of the task but also \emph{monotonically} related to the number of tasks the agent has been assigned to before. The sum-of-completion-time problem already shows how such a positional effect can be naturally derived from the optimization objective of minimizing total completion time. Moreover, the positional cost may alternatively come from an abstraction of the varying \emph{efficiency} when an agent works in reality. For example, the so-called \emph{learning effect}~\cite{positional_learning:1999} \cite{positional_learning:2000} usually helps a human agent to work more and more efficiently, while on the other hand various \emph{deteriorating effects}~\cite{time_deteriorating:1990} may just do the opposite (e.g. human/animals can get tired during working, devices can wear off during the usage, or the working situation is just getting worse over time in medical treatment and diseaster rescue). See Appendix \ref{sec_battery} for a motivating application for the authors, where the goal is to optimize the overall battery efficiency in electrical systems powered by multiple batteries, in light of the phenomenon that the power efficiency of battery gets worse as the battery is discharging. 

In cases where the tasks can be run in arbitrary order, we know from the \emph{rearrangement inequalities} by Hardy et al.~\cite{inequality:1952} that the problem can be solved by a simple \emph{Shortest-Processing-Time} (SPT) rule that always matches the task with the shortest processing time to the position with the largest weight. However, when the \emph{non-reorderable} cosntraint is imposed, which could be either caused by a priority of the tasks (such as in a queue) or by the online nature of the problem, the greedy SPT rule becomes suboptimal (see Figure \ref{fig:counter_example} in Appendix \ref{sec_counter_example} for a counter-example), and there seems to be no obvious way to solve the problem in polynomial time.  

In this paper we propose a simple but nontrivial algorithm that runs in $O(k^2 n)$ time, and we show that this algorithm is optimal under any problem instance with two task-types. The algorithm works by making greedy decision for each task \emph{sequentially} based on stopping thresholds in a ``greedy-like" allocation simulation. We expect the combinatorial structures of the problem exhibited by our algorithm can inspire the design of practical and optimal algorithms in more general settings of this important problem. 


%% file: sec_related.tex
\subsection{Connections with Related Work}

Both sum-of-completion-time (or, the \emph{mean flow time}) and max-of-completion-time (i.e. the \emph{makespan}) are extensively studied optimization objectives. It is widely known that the makespan optimization problem is strongly NP-hard even assuming a constant positional function $f(p) \equiv 1$, both for its uncapacitated version (i.e.,  multiprocessor scheduling/bin packing) and capacitated version (i.e., the 3-partition problem). Many studies were thus focusing on designing asymptotic-PTAS or constant-ratio approximation algorithms for makespan optimizations, especially under generalized cost functions~\cite{bp:1994} \cite{bp:2012} \cite{bp:2006} \cite{bp:2008}. The objective functions of most these generalized cost models are \emph{symmetric} with respect to the tasks/items, and thus the order of items has no impact on the aggregate value. In contrast, in the positional allocation problem the ``real cost" of an item further depends on \emph{where} it is put in the bin. 

Meanwhiles, another line of research tries to find polynomial-time exact algorithms for the makespan optimization problem assuming a constant number of task/item types. Specifically, Leung \cite{bp-c:1982} presented an $O(n^{2c-2} \log k)$ dynamic programming algorithm, where $c$ is the number of item types. However, since problems in such a setting can admit compact inputs (in fact, encoded by only $2c$ numbers), the running time of a polynomial algorithm for such \emph{high-multiplicity} problems needs to be polynomial to $\log n$ (rather than to $n$). A polynomial algorithm for the case $c=2$ was first given by McCormick, Smallwood and Spieksma in \cite{bp-c:1997}. Later, Eisenbrand and Shmonin~\cite{bp-c:2006} showed that actually only $3$ different ``packing ways" are needed for the case of $c=2$. Very recently, Goemans and Rothvo{\ss}~\cite{bp-c:2014} extended the techniques in \cite{bp-c:2006} and gave the first polynomial-time bin packing algorithm for arbitrary (but constant) number of item types. Our work pursuits similar goals with these works, trying to find polynomial-time algorithms for problem instances with constant number of task/item types, albeit with a different optimization objective encompassing flow-time. Essentially, the metrics of makespan and flow-time correspond to the $L_{inf}$ norm and $L_1$ norm, respectively (see Eq. \ref{eq:total_cost} in Section \ref{sec_preliminary}). Besides, also note that the order of tasks in the input sequence plays a crucial role in the positional allocation problem considered in this paper, which means binary instances of this problem cannot be compressed into a sequence of multiplicities, but will have the same input format as the general form of the problem, thus having $O(n)$ input length.

On the other hand, scheduling under positional costs is also an active area in operations research. Biskup and others~\cite{positional_learning:1999} \cite{positional_learning:2000} first considered the learning effects in single-machine scheduling, in which the positional weights decrease with the positions, typically modeled by explicit polynomial functions in most later works~\cite{positional_learning:2001}. See \cite{positional_learning:2008} for a survey of them. Browne and Yechiali~\cite{time_deteriorating:1990} first introduced the deteriorating effects in scheduling problems, in which the positional weights increase with the positions, typically modeled by explicit polynomial functions~\cite{positional_deteriorating:2005} or exponential functions~\cite{positional_deteriorating:2008}. See \cite{positional:2012} for a survey of works in this line. Besides, some works also considered the parallelel machine scheduling problems with positional costs~\cite{positional:2001}  \cite{positional:2012}. In most of the positional scheduling works presented above, the key is to find good \emph{permutations} of the task sequence so as to minimize the objectives considered, and the classic Short-Processing-Time-first (SPT) rule turns out to be optimal in various settings (e.g. see the summary table in \cite{positional_learning:2008}). Differently, the positional allocation problem considered here imposes a strict \emph{nonorderable} constraint on the order of tasks. Note that the constraint is different from classic \emph{precedence constriants} in that, the former only constrains the \emph{order} of tasks in the same machine (but it is possible to run tasks in parallel in different machines) while the latter further rules out any parallelism between tasks with precedence relationship. It turns out that the nonorderable constraint invalidates the mostly-used SPT rule. Actually, as demostrated later, the positional allocation problem studied in this paper exhibits a quite different combinatorial structure, which leads to practical optimal algorithms quite different with the SPT rule.

Finally, there is an interesting connection between the positional allocation problem and the \emph{Baysian Multi-Armed Bandit} (MAB) problem . In the Baysian MAB problem, we are given $k$ ``bandit-arms", and each arm $r$ is in an observable ``state" $\Q{r}{t}$ at time-slot $t$. In any round $t$, we are asked to choose one arm $x_r$, leading to a stochastic payoff $\alpha^t \cdot R(\Q{x_t}{t})$, and also causing the chosen arm $x_t$ to stochastically change its state to $\Q{x_t}{t+1} = P(\Q{x_t}{t})$. Actually, the function $R$ corresponds to a parameterized probability distribution of payoff, the state $\Q{r}{t}$ corresponds to the parameter setting of $R$ that the player ``believes" the arm $r$ ``should be in" at time $t$, and the state transition function $P$ follows the Baysian inference principle for probability distribution $R$. The goal is to maximize the cumulative reward at infinite horizon. It is not hard to see that the Baysian MAB problem is essentially a stochastic version of the positional allocation problem, with a special task sequence $\{T_t = \alpha^t\}$ on one hand, while with general stochastic transition functions on the other hand (in the positional allocation problem, $P(\Q{r}{t}) = \Q{r}{t}+1$, and $R(\Q{r}{t})=0$ when $\Q{r}{t} > m$). In 1979, Gittins found an elegant \emph{simulation-based} algorithm~\cite{gittins:1979}, which is proven to be optimal for the Baysian MAB problem~\cite{gittins:1999}. Describing in our language, given an infinite task sequence $\alpha^1,\alpha^2,\alpha^3,\dots$, where $0<\alpha<1$, the algorithm assigns a score $V_r$ to each agent $r$ according \emph{solely} to the current capacity $m_r$ of that agent (i.e., $V_r$ is independent to any other $m_{r'\neq r}$), then the algorithm simply allocates the first task $T_1$ to the agent with the highest score. The score, later called the \emph{Gittins Index}, happens to correspond to the (expected) normalized cumulative cost of the \emph{optimal stopping strategy} of a simulation to keep allocating tasks in the single machine $r$. Note that the task sequence in Baysian MAB is by default sorted in a Shortest-Processing-Time-first manner~\footnote{More accurately, the Baysian MAB problem, as a reward maximization problem, always receives exponentially decreasing inputs, which translates to an increasing (SPT-first) task sequence in its dual problem concerning cost minimization.}, and due to the monotonicity of the positional function, the Gittins-Index-based algorithm degenerates to the naive SPT algorithm in the positional allocation problem, which is known to be suboptimal in general (but indeed optimal for that specific single task sequence!). Interestingly, as shown in this paper, it turns out that the truely optimal algorithm for the positional allocation problem may still exhibit a very similar working pattern with the Gittins' algorithm (at least for binary inputs), namely that the optimal decision for each task can be made \emph{sequentially} based on the ``stopping threshold" of a ``greedy-like" allocation simulation.

%% file: sec_preliminaries.tex
\newcommand{\aaa}[2]{a^{#1}_{#2}}
\newcommand{\aaA}[1]{\mathbf{a}^{#1}}
\def\T{{\rm T}}

\section{Preliminaries} \label{sec_preliminary}


\begin{definition} \label{def:allocation}
A \emph{$k$-allocation scheme of $n$ tasks} is a partitioning of the sequence $1,2,\dots,n$ into $k$ subsequences, denoted by $\mathbf{A}_{n,k} = (\aaA{1}, \aaA{2}, \dots, \aaA{k})$, where (1) $\aaa{r}{i} \neq \aaa{s}{j}$ if $(r,i) \neq (s,j)$ (disjointness); (2) $\sum_{r=1}^k |\aaA{r}| = n$ (completeness); (3) $\aaa{r}{i} < \aaa{r}{j}$ if $i < j$, for any $1\leq r \leq k$ (monotonicity).
\end{definition}

The \emph{size} of an allocation scheme $\mathbf{A}_{n,k}$ is $m$ if $|\aaA{r}|=m$ for every $\aaA{r} \in A_{n,k}$. Note that only an allocation scheme with uniform cardinality has a well-defined size. The \emph{position} of task $t$ under allocation scheme $\mathbf{A}_{n,k}$ is $p_t$ if there exists $r \in [k]$ such that $\aaa{r}{p_t}=t$ . Given an allocation scheme $\mathbf{A}_{n,k}$, any integer sequence $T_1, T_2, \dots, T_n$ can be accordingly partitioned into $k$ subsequences, denoted by $\mathbf{T}_{n,k} = (\{T_{\aaa{1}{p}}\}, \{T_{\aaa{2}{p}}\},$ $ \dots, \{T_{\aaa{k}{p}}\})$. For convenience we will write $T_{r,p}$ for $T_{\aaa{r}{p}}$ when the context is clear.

\vspace{0.1in} \noindent \textbf{Positional Allocation Problem. } 
Given a problem instance $\mathbf{I}=(n,k,m,$ $f,$ $T_1, \dots, T_n)$ where $n = k m$, $T_i \in \mathbb{N}$, and $f : \mathbb{N} \mapsto \mathbb{N}$ is arbitrary monotonically decreasing function. We want to find an allocation scheme $A_{n,k}$ of size $m$ in order to  
\begin{equation} \label{eq:total_cost}
\min_{\mathbf{A}_{n,k}} 	 ~~~ \left|\left| 
\left( \begin{array}{cccc}
T_{1,1} & T_{1,2} & \ldots & T_{1,m} \\
\vdots & \vdots & \vdots \\
T_{k,1} & T_{k,2} & \ldots & T_{k,m}
\end{array} \right) 
\cdot
\left( \begin{array}{c} f(1) \\ f(2) \\ \vdots \\ f(m) \end{array} \right)
\right|\right|_1
\end{equation}

In the above problem formulation, a positional allocation algorithm outputs an allocation scheme $\mathbf{A}_{n,k}$ for a given problem instance. Equivalently, a positional allocation algorithm may output a \emph{decision sequence} $\x = (x_1,\dots, x_n)$ where each $x_t \in [k]$ denotes the index of the agent assigned to the task $t$. Clearly, any decision sequence corresponds to a unique allocation scheme. The following lemma shows that the reverse is also true: any given allocation scheme also corresponds to a unique decision sequence. In other words, an allocation scheme is equivalent to a decision sequence. See the proof in Appendix \ref{proof:allocation}.
\begin{lemma} \label{lemma:allocation}
For any allocation scheme $\mathbf{A}_{n,k} = (\aaA{1},$ $\aaA{2}, \dots, \aaA{k})$, there exists a unique decision sequence $\mathbf{x}=(x_1, \dots, x_n)$ such that $\aaa{x_t}{p_t} = t$ for each $1\leq t \leq n$.
\end{lemma}

In the rest of this paper we will mainly discuss algorithms assuming they output decision sequences, and as a general technique, we often prove the sub-optimality of a given decision sequence $\x$ by re-arranging some tasks in the allocation scheme corresponding to $\x$ and showing that the re-arranged allocation scheme has lower cost than the original one. In such a proof, the monotonicity property of allocation scheme is the key to guarantee that the re-arranged allocation scheme is still valid (i.e. ``achievable" by some decision sequence). 

%
\vspace{0.1in}

We remark that our problem formulation encompasses some other related models. For example, although in our formulation each agent must be assigned \emph{exactly} $n/k$ tasks, both the problem variants with and without cardinality constraints (in which an agent $r$ can run \emph{at most} $m$ tasks and \emph{arbitrary} number of tasks, respectively) can be reduced to the problem formulated here by appending enough number of ``null tasks" with $T_i=0$. Furthermore, although our formulation assumes the positional cost function $f$ is decreasing, a problem instance with increasing function can be reduced to an instance in our model by reversing the task sequence. Specifically, for any instance $\mathbf{I}=(n,k,m,g,T_1, T_2, \dots, T_n)$ with increasing function $g$, we can construct an instance $\mathbf{I}'=(n,k,m,f,T_n,\dots,T_2,T_1)$ with the decreasing function $f(p)=g(m+1-p)$, and from Lemma \ref{lemma:allocation} we know that: if $\x' = (x'_1,x'_2, \dots, x'_n)$ is the solution of $\mathbf{I}'$ in our model, then $\x=(x'_n, \dots, x'_2, x'_1)$ is the solution of $I$ in the model with increasing positional function (and vise versa). Similarly, the same reduction also works for the problem variant with decreasing positional function but reversely-growing positional index (from $m$ back to $1$). Also note that the tricks presented above can be further combined together to reduce more combinations of variants to our problem. For example, let $\mathbf{I}=(n,k,T_1, T_2, \dots, T_n)$ be an instance of the sum-of-completion-time problem presented at the beginning of the paper. The ``equivalent instance" of $\mathbf{I}$ in our model is $\mathbf{I}' = (n,k,n,f, T_n, \dots, T_2, T_1, 0,0,\dots,0)$ where there are $(kn-n)$ ``null tasks" in $\mathbf{I}'$ and $f(p)=n+1-p$. 

Finally, the current probem formulation is presented in a form for the sake of simplicity, and the algorithms presented in this paper may apply to some natural generalizations of the problem. For example, in our formulation every agent is assigned with the same number of tasks, while our algorithmic results also apply to problems with arbitrary capacity plan $(m_1, m_2, \dots, m_k)$ in which agent $r$ may be assigned with (exactly or at most) a different number of $m_r$ tasks. Besides, in this paper we couple the position weights and task-specific costs with the multiplication operator (see Eq. \ref{eq:positional} and Eq. \ref{eq:total_cost}), while the algorithmic discussions also apply to the more general cases where the cost function $f(p_i,T_i)$ is monotone and  has positive mixed partial derivaties $\frac{\partial^2 f}{\partial p \partial T} = \frac{\partial^2 f}{\partial T \partial p} > 0$. 

\vspace{0.1in} \noindent \textbf{The Shortest-Processing-Time Rule. } 
The monotonicity of the positional cost function implies a simple ``principle" to allocate tasks: In general, we tend to run tasks with relatively smaller costs first (thus coupled with larger positional weights) while to allocate tasks with relatively larger costs later (thus coupled with smaller positional weights). Lemma \ref{lemma:exchange} justifies this intuition formally.
\begin{lemma} \label{lemma:exchange}
(Rearrangement Inequality \cite{inequality:1952}) If a monotonically increasing function $f(Q, T) : \Rnn \times \Rnn \mapsto \Rnn$ has positive mixed derivatives $\frac{\partial^2 f(Q,T)}{\partial Q \partial T} = \frac{\partial^2 f(Q,T)}{\partial T \partial Q} > 0$, then for any $Q_1 \geq Q_2$, and $T_1 \geq T_2$, we have 
\begin{equation}
f(Q_1, T_1) + f(Q_2, T_2) \geq f(Q_1, T_2) + f(Q_2, T_1)
\end{equation}
\end{lemma}

The above principle suggests that, if we \emph{could} arbitrarily change the order of the tasks, the optimal algorithm would be simply to sort the tasks in ascending order and to assign them among the agents in a round-robin way. However, the nonreorderable constraint of the problem brings additional difficulty that we are forced to allocate tasks sequentially. In this case, a naive greedy allocation algorithm may be to couple a large (small) task with the smallest (largest) positional weight \emph{at that time}. This naive greedy algorithm turns out to be sub-optimal. For example, Figure \ref{fig:counter_example} in the appendix shows an instance with binary costs ($k=4$ and $n=80$), where the allocation scheme of the naive greedy algorithm (the left side) has larger cost than another allocation scheme (the right side). Actually, the better allocation scheme at the right side is optimal under this instance, and it comes from a simple and efficient algorithm proposed in the next section. In Section \ref{sec:optimality} we prove that this algorithm minimizes the total cost for any binary-valued instance of the positional allocation problem.


%% file: sec_algorithm.tex
\section{The Algorithm}
We assume $T_i \in \{H,L\}$, $0\leq H<L$ in this section. A pragmatic motivation of the assumption is that, often in real world the task-specific costs follow a \emph{bimodal} distribution, in which case a two-value separation may well approximate the real values. For example, in the multi-battery application presented in Appendix \ref{sec_battery}, the power consumption of a device can greatly depend on whether the device is active or on standby, with a huge gap of more than 100x (See Figure \ref{power_gap}). Moreover, detecting the active-standby mode of a device is usually much more efficient, robust, and easier than measuring the exact value of the power consumption of the device.

A straightforward $k$-dimensional dynamic programming procedure can solve the positional scheduling problem. Specifically, define $\mathbf{m}=(m_1,m_2, \dots, m_k)$ as the state vector for the situation where there are $m_r$ ``slots" left in each agent $r$. Given a task sequence $\mathbf{T} = (T_1,T_2, \dots, T_{\sum m_r})$, define $COST_{\mathbf{T}}(\mathbf{m})$ as the minimum total cost over all possible decision sequences that matches $\mathbf{T}$ to $\mathbf{m}$. By definition of the positional allocation problem we have
\begin{equation*}
COST_{\mathbf{T}}(\mathbf{m}) = \min_{r \in[k]}  \left\{~ COST_{\mathbf{T} \setminus T_1}(m_1,\dots, m_r-1, \dots, m_k) + f(m-m_r) \cdot T_1 ~\right\}
.\end{equation*}
Since the size of the state space is no more than $n^k$, the time complexity of the dynamic programming procedure is polynomial to $n$ for constant $k$. For general $k$, however, the size of the state space is at least the \emph{integer partition function} $p(\frac{n}{2}) = 2^{\Theta(\sqrt{n})}$ (e.g., when $k=\frac{n}{2}$), which is super-polynomial to the input size. Meanwhiles, another drawback of the dynamic programming solution is that it may not easily adapte to the cases where the information of the tasks is limited, such as in online allocation scenarios where we (at best) only know a stochastic generating process of the workload, rather than a deterministic task sequence.

\begin{algorithm}[t] \label{algo:greedy}
\caption{The basic version of the simulation-based algorithm}

\SetKw{Continue}{continue}
\SetKwFunction{TA}{ThresholdAllocation} 

\KwIn{$ m_1 \dots m_k, T_1 \dots T_n$, where $m_1\leq m_2 \leq \dots \leq m_k$}
\KwOut{$x_1 \dots x_n$}
\BlankLine

set $\Q{r}{1} = m_r$ for each $1 \leq r \leq k$

\For{$t = 1$ \KwTo $n$}
{
	$x_t \gets$ \TA ($\Q{1}{t} \dots \Q{k}{t}, T_t \dots T_n$)  

	$\Q{r}{t+1} \gets \Q{r}{t} -\mathds{1}(x_t = r)$ ~~for each $1 \leq r \leq k$
}

\Return $\mathbf{x}$  

\BlankLine
\BlankLine
\BlankLine

\KwSty{Function} \TA \label{algo:greedy_single}

\KwIn{$ m_1, \dots, m_k, T_1, \dots, T_n$}
\KwOut{the agent to which $T_1$ is assigned}

	\For{$\gamma = k$ \KwTo $1$}
	{
		\If{$m_{\gamma} = m_{\gamma-1}$}
		{  
			\Continue   \label{line:stable2}
		}

		\For{$h = \gamma$ \KwTo $k$}
		{
			$S_L = \{ 1,2,\dots,h-1 \}$  
			
			$S_H = \{ \gamma,\gamma+1, \dots, h \}$  

			$Z_L \gets \sum_{i \in S_L} \min \{m_i, m_{\gamma-1}\}$  

			$Z_H \gets \sum_{i \in S_H} (m_i - m_{\gamma-1})$  
			
			\BlankLine

			$\mathcal{L} = \{i: i \in \{1, \dots, Z_L+Z_H\} \text{~~and~~} T_i > T_1\}$  

			\If{$|\mathcal{L}| \geq Z_L$} 
			{
				\Return $\gamma$  \label{line:question}
			}
		}
	}
\end{algorithm}

 \begin{figure*} [t]
	\begin{centering}
	\includegraphics[width=1.0 \textwidth]{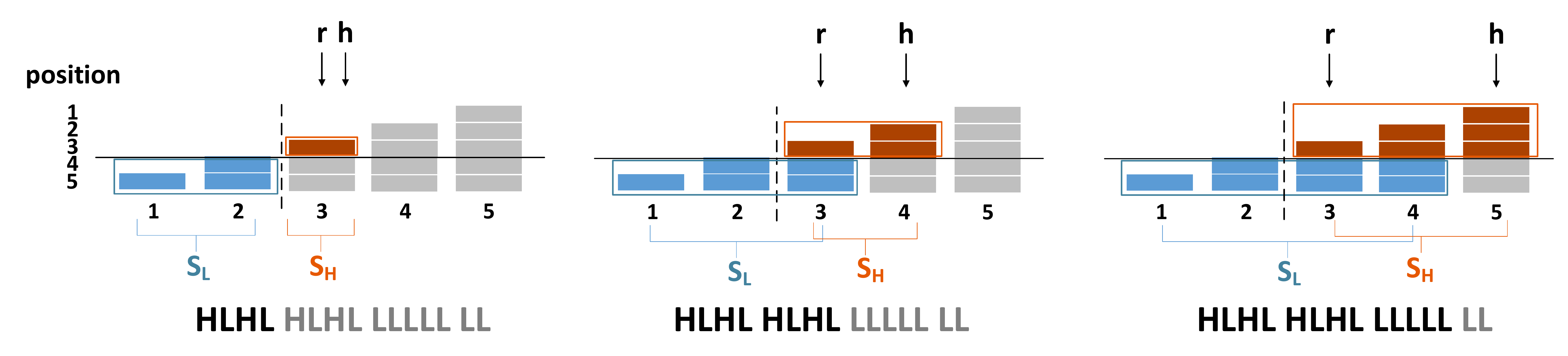}
	\caption{Illustration of a simulation process performed by the ThresholdAllocation routine of Algorithm \ref{algo:greedy} for $r=3$ under the task sequence $HLHL~HLHL~LLLLL~LL$ ($n=15$, $k=5$). The three diagrams correspond to the initial setup and two later expasions of H-zone (positions in red) and L-zone (positions in blue) in the simulation, respectively. In particular, when $h=5$ (the right side), the H-zone is of size $Z_H=6$ and the L-zone is of size $Z_L=7$. Since there are at least $7$ ``L" tasks in the first $7+6=13$ tasks, the algorithm will assign the first task to agent $r=3$.}
	\label{fig:algorithm}
	\end{centering}
\end{figure*}

In this section, we present a simple algorithm that turns out to generate the optimal decision sequence for any problem instance with binary-valued task sequence. The basic version of the algorithm is shown in Algorithm \ref{algo:greedy}. Again we use $(m_1, m_2,\dots, m_k)$ to denote the cardinality capacities of agents, and without loss of generality assume $m_1\leq m_2 \leq \dots \leq m_k$. To allocate the first task $T_1$, the algorithm iterates over each resource $r$ in the descending order (priority) to decide whether to put $T_1$ in $r$. The decision for each $r$ is made via doing a ``simulated allocation" as follows (Line $12 \sim 21$ in Algorithm \ref{algo:greedy}): the simulation first setups a \emph{L-zone} and a \emph{H-zone} among all the available positions, then sequentially allocates the task sequence, starting from $T_1$, by sending all ``large" tasks (i.e. the ones with $T_i >T_1$) in the L-zone and all ``small" tasks (the ones with $T_i \leq T_1$) in the H-zone. Whenever the ``small" tasks overflow from the H-zone, both the H-zone and the L-zone expand, and the simulated allocation continues. If at any time in the simulation the L-zone is filled up by ``large" tasks (is full), the algorithm immediately stops the simulation and allocates $T_1$ to agent $r$; otherwise it will choose some other agent with smaller id in the later simulations. Note that the algorithm will guarantee to choose agent $1$ (which has the smallest capacity) if getting chance to run simulation on it (i.e. with $r=1$). Also note that the task with the largest cost $\max T_i$ (i.e. L in this context) will always be assigned to the smallest non-empty agent.

The initialization and expansions of H-zone and L-zone are based on the current capacities of agents $(\Q{1}{t}, \dots ,\Q{k}{t})$ and the index $r$ of the (agent) candidate on which the simulation focuses. Let $p^*$ be the ``current position" of agent $r-1$ (so $p^*=m+1-\Q{r-1}{t}$). At any time of the simulation the H-zone only contains positions smaller than $p^*$ while the L-zone only contains positions equal or larger than $p^*$. Initially, the H-zone contains all such legal positions for the agent $r$ (i.e., the single position of $p^*$ in agent $r$), and the L-zone contains all legal positions for agents smaller than $r$. In every expansion, both the H-zone and L-zone include legal positions of one more agent, from $r+1$ to $k$. Figure \ref{fig:algorithm} illustrates an example simulation process with five agents and $r=3$. 

Given a H-zone and L-zone setup, to check whether the L-zone can be filled up before the H-zone overflows, the algorithm only needs to count whether the number of ``large" tasks within a look-ahead ``window" exceeds a ``threshold". Specifically, Let $Z_H$ and $Z_L$ be the sizes of H-zone and L-zone respectively, the L-zone can be ``successfully" filled up if and only if there are at least $Z_L$ ``large" tasks in the first $Z_H+Z_L$ tasks. Suppose task $t$ is the head of the current task sequence, the last task to check in the look-ahead will be $Z_L+Z_H+t-1$, which will be called the \emph{look-ahead horizon from task $t$ under setup $(r,h)$} in the rest of the paper, or just \emph{horizon} when the context is clear. 

Note that the output of Algorithm \ref{algo:greedy} does not depend on the specific form of the positional function $f(p)$ at all, nor on the specific values of the task $T_i$, nor even on the detailed pattern of how the tasks are arranged within each look-ahead window. Moreover, the following lemmas establish some monotonicity properties of Algorithm \ref{algo:greedy}. Specifically, Lemma \ref{lemma:capacity} asserts that the algorithm never changes the order of the capacities of agents, i.e., the monotonicity of agent-capacities is conserved; Lemma \ref{lemma:horizon} asserts that the look-ahead horizons under any given setup $(r,h)$ always move forward in the direction from $1$ to $n$, i.e., the monotonicity of horizon for given look-ahead setup is conserved; and Lemma \ref{lemma:decision} asserts that the allocation decisions for the same task type always move forward in the direction from agent $1$ to $k$, i.e. the monotonicity of agent id's assigned to given task-type is conserved.

\begin{lemma} \label{lemma:capacity}
Under any problem instance $I = (m_1 \dots m_k, T_1 \dots T_k)$ with two types of task, $m_1 \leq m_2 \leq \dots \leq m_k$, for the agent-capacity variables $\{\Q{r}{t}\}$ computed by Algorithm \ref{algo:greedy} under $I$, we have $\Q{r}{t} \leq \Q{s}{t}$ if $r < s$, for any $1\leq t \leq n$.
\end{lemma}

\begin{lemma} \label{lemma:horizon}
Under any problem instance $I = (m_1 \dots m_k, T_1 \dots T_k)$ with two types of task, $m_1 \leq m_2 \leq \dots \leq m_k$, for any given look-ahead setup $(r,h)$, $1\leq r \leq h \leq k$, let $Z^{(r,h,i)}_H$ and $Z^{(r,h,i)}_L$ be the sizes of H-zone and L-zone (respectively) when Algorithm \ref{algo:greedy} is allocating task $i$ under setup $(r,h)$, and let $Z^{(r,h,j)}_H$ and $Z^{(r,h,j)}_L$ be the sizes of H-zone and L-zone (respectively) when Algorithm \ref{algo:greedy} is allocating task $j$ under the same setup $(r,h)$,  we have 
\begin{equation}
Z^{(r,h,i)}_H + Z^{(r,h,i)}_L + i -1 \leq Z^{(r,h,j)}_H + Z^{(r,h,j)}_L + j -1 \text{~~~~if~~~} i < j
.\end{equation}
\end{lemma}

\begin{lemma} \label{lemma:decision}
Under any problem instance $I = (m_1 \dots m_k, T_1 \dots T_k)$ with two types of task, $m_1 \leq m_2 \leq \dots \leq m_k$, let $i$ and $j$ be two tasks of the same type (i.e. $T_i = T_j$), and let $(x_1, \dots, x_n)$ be the decision sequence output by Algorithm \ref{algo:greedy} under $I$, we have $x_i \leq x_j$ if $i < j$.
\end{lemma}

In particular, thanks to Lemma \ref{lemma:horizon}, all data variables (the horizons, thresholds, and counters) used by Algorithm \ref{algo:greedy} can be updated incrementally, yielding a more efficient implementation as shown by Algorithm \ref{algo:threshold} in Appendix \ref{sec_faster_algo}. It is easy to see that $O(k^2 n)$ time and $O(k^2)$ space is sufficient for Algorithm \ref{algo:threshold} to compute the same decision sequence with Algorithm \ref{algo:greedy} for any problem instance with two types of task. 

\begin{theorem} \label{thm:threshold}
Under any problem instance $I = (m_1 \dots m_k, T_1 \dots T_k)$, $T_i \in \{H,L\}$ for $1\leq i \leq n$, $0\leq H<L$, Algorithm \ref{algo:threshold} returns the same decision sequence with Algorithm \ref{algo:greedy} in $O(n \cdot k^2)$ time and with $O(k^2)$ space.
\end{theorem}

%% file: sec_optimality.tex
\section{Optimality of Algorithm \ref{algo:greedy}} \label{sec:optimality}

In this section we prove that Algorithm \ref{algo:greedy} minimizes the total cost of any problem instance of positional allocation with two types of task. The basic idea is to show that if at any time we don't assign a task $t$ to the agent decided by Algorithm \ref{algo:greedy}, from the resulting decision sequence $\x$ we can always construct another sequence $\x^*$ such that $\x^*$ follows Algorithm \ref{algo:greedy} on $t$ and has less total cost than $\x$. The construction is by re-arranging some tasks in the corresponding allocation scheme of $\x$ without violating the monotonicity property of Definition \ref{def:allocation}.

Since Algorithm \ref{algo:greedy} runs in an iterative manner, we only prove this for the first task $T_1$. In order to be consistent with the input of Function ThresholdAllocation of Algorithm \ref{algo:greedy}, we slightly generalize the problem formulation to allow variable-sized agents in this section as follows: 

\noindent \textbf{Variable-Sized Position Allocation (VSPA) Problem.} Given a problem instance $\mathbf{I}=(m_1,\dots,m_k,T_1,\dots, T_n)$, where $T_i$ is the intrinsic-cost of task $i$ and $m_r$ is the cardinality capacity of agent $r$, without loss of generality assume all agents are ``non-empty" at the beginning, and is sorted increasingly in capacity, that is, $0 < m_1 \leq m_2 \dots \leq m_k$. A valid allocation scheme $\mathbf{A} = (\aaA{1},\dots, \aaA{k})$ is required to be consistent with the cardinality constriants, that is, $|\aaA{r}|=m_r$ for each $r \in [k]$. The allocation scheme is indexed by the capacity of the agent when the task is assigned to (i.e. the reversed position). Instead of directly applying the positional function $f(p)$, we define 
\begin{equation} \label{eq:capacity_function}
g(q) = f(m+1-q)
.\end{equation}
Assigning task $i$ to agent $r$ of capacity $m_r$ causes a cost of $g(m_r) \cdot T_i$. Again, every decision sequence $\x$ corresponds to a unique partitioning of the task sequence $\{T_{r,q}\}$, and the goal is to minimize the total cost defined by  
\begin{equation} \label{eq:total_cost_2}
COST_{\mathbf{I}}(\x) = \sum_{r=1}^k \sum_{q=1}^{m_r} g(q) \cdot T_{r,q}  
.\end{equation}
One can verify that the VSPA problem is exactly the original positional allocation problem when $m_1=m_2=\dots=m_k=m$.

Now we will prove the optimality of Algorithm \ref{algo:greedy} for VSPA instances with two types of task. The main part of the proof is separated in cases of $T_1 = L$ (Lemma \ref{lemma:low_greedy}) and $T_1=H$ (Lemma \ref{lemma:high_greedy}). In the former case, a straightforward re-arrangment can be done by observing the fact that Algorithm \ref{algo:greedy} always assigns a L task to the ``smallest" agent (i.e., in agent $1$). See the proof in Appendix \ref{proof_low_greedy}. 

\begin{lemma} \label{lemma:low_greedy}
If $I=(m_1 \dots m_k, T_1 \dots T_n)$ is an instance of the VSPA problem with $T_i \in \{H,L\}$, $0\leq H<L$, $0 < m_1 \leq m_2 \leq \dots \leq m_k$, and $T_1 = L$, then for any decision sequence $\x = (x_1,x_2,\dots,x_n)$ with $x_1 > 1$, there exists another decision sequence  $\x^*=(1, x^*_2,\dots,x^*_{n})$ such that $COST_I(\mathbf{x^*}) \leq COST_I(\mathbf{x})$.
\end{lemma}

The proof for cases with $T_1 = H$ (Lemma \ref{lemma:high_greedy}) is by induction. Specifically, assume by induction that Algorithm \ref{algo:greedy} is optimal for the subsequent tasks $T_2, \dots, T_n$, which yields a specific sequence $\x$ that is guaranteed to be optimal among the set of decision sequences \emph{not} following Algorithm \ref{algo:greedy} at $T_1$. We will show that the specific pattern of $\x$ always enables a re-arrangement to beat itself, and thus beat any sequence not following Algorithm \ref{algo:greedy} at $T_1$. More specifically, suppose Algorithm \ref{algo:greedy} assigns $T_1$ to agent $\gamma^*$, there can be two ways to \emph{not follow} this decision: i) to allocate $T_1$ ``lower", in some agent $r < \gamma^*$; or ii) to allocate $T_1$ ``higher", in some agent $r > \gamma^*$. The constructions of better decision sequences are further separated into these two cases. See the proof in Appendix \ref{proof_high_greedy}.

\begin{lemma} \label{lemma:high_greedy}
Suppose $I=(m_1 \dots m_k, T_1 \dots T_n)$ is an problem instance with $T_i \in \{H,L\}$, $0\leq H<L$, $0 < m_1 \leq m_2 \leq \dots \leq m_k$, and $T_1 = H$, and suppose Algorithm \ref{algo:greedy} assigns task $1$ to agent $\gamma^*$ under $I$, then for any decision sequence $\x = (x_1,x_2,\dots,x_n)$ with $x_1 \neq \gamma^*$, there exists another decision sequence $\x^*=(\gamma^*, x^*_2 \dots,x^*_{n})$ such that $COST_I(\mathbf{x^*}) \leq COST_I(\mathbf{x})$.
\end{lemma}

Finally, Theorem \ref{thm:optimality} combines results proved in Lemma \ref{lemma:low_greedy} and Lemma \ref{lemma:high_greedy} to complete the proof of the instance-optimality of Algorithm under binary instances consisting of $H$ and $L$.

\begin{theorem} \label{thm:optimality}
Algorithm \ref{algo:greedy} minimizes the total cost defined by Eq. \ref{eq:total_cost} for any instance of the positional allocation problem with two types of task.
\end{theorem}

%% file: sec_appendix.tex
\section{The Multi-Battery Problem} \label{sec_battery}
Energy-efficiency is a key concern on mobile devices that depend on batteries to provide the power required to maintain operation. Whenever a device draws power, not all energy that is drawn from the battery is actually useful in the sense that it ends up powering the device. In fact, only a fraction of the energy drained from the battery ends up powering the device, the remainder is wasted: it heats up the device. We call these two components of the energy drained from a battery with each load the \emph{useful energy} and the \emph{wasted energy}, respectively. There are two main factors that determine the amount of wasted energy in the battery in a given time duration: the power consumption of the current user load and the \emph{internal resistance} of the battery. In turn, one of the key factors determining internal resistance is the \emph{State-of-Charge} (SoC) of the battery, i.e., how much remaining charge is in the battery. The quantatitive relationship between the wasted energy $E_{waste}$, the power of load $p$, and the SoC of battery $s$ can be approximated by the following formula:
\begin{equation} \begin{aligned}
E_{waste} &= T \cdot \left( 1-\sqrt{1-\frac{2r}{V_{oc}} \cdot p} ~\right) / \frac{2r}{V_{oc}} \\
r &= R_0 + \alpha \cdot (1-s) 
,\end{aligned} \end{equation}
where $V_{oc}$ as the open-circuit voltage, $R_0$ as the initial resistance of the battery, $\alpha$ as the DCIR-SoC coefficient, and $T$ as the time length of the load -- all can be considered to be constants. One can verify that $E_{waste}$ is increasingly monotone to $p$ and decreasingly monotone to $s$. Also, the function $E_{waste}(p,s)$ always has positive mixed partial derivatives. By Lemma \ref{lemma:exchange} we know that, in general, we should try to power low-power loads by batteries in relatively low state-of-charge and power high-power loads by batteries in relatively high state-of-charge, so as to minimize the total wasted energy during a battery discharging cycle.
 \begin{figure*} [h!]
	\begin{centering}
	\includegraphics[width=1.0 \textwidth]{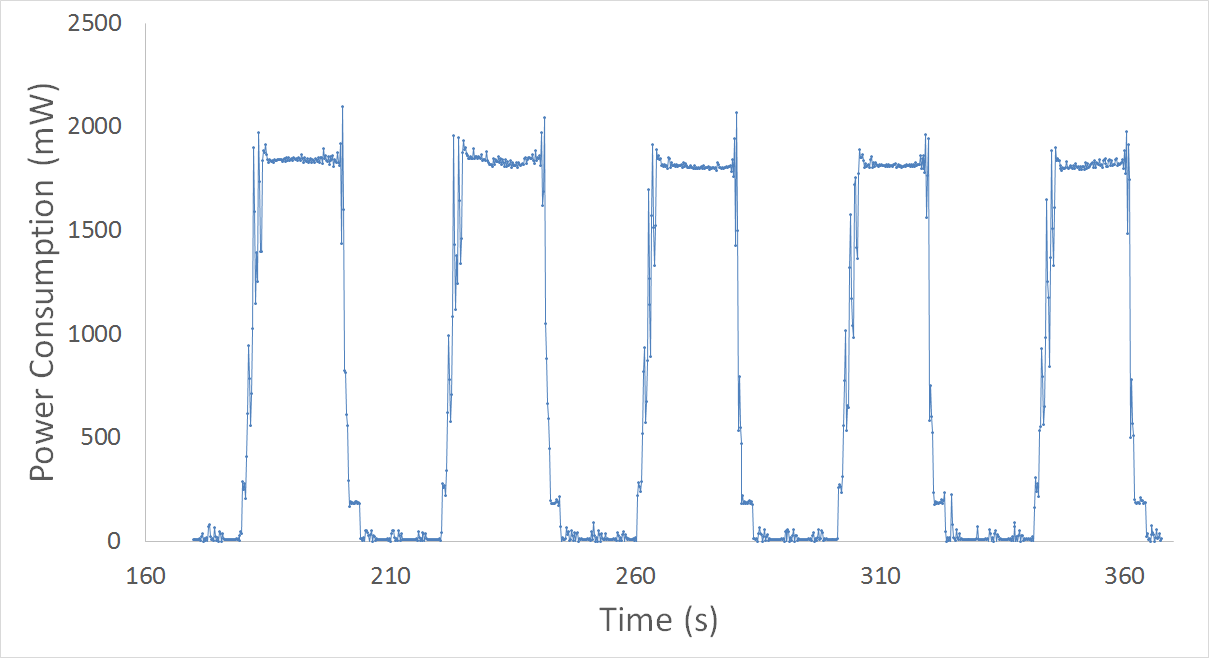}
	\caption{Power consumption curve of smart phones. The average power under active mode is about $1800$ mW, while the number under standby mode is less than $16$ mW -- a gap of more than 100x.}
	\label{fig:power_gap}
	\end{centering}
\end{figure*}

\section{Example Instance Showing the Naive Greedy Strategy is Suboptimal} \label{sec_counter_example}

The naive greedy strategy for the positional allocation problem is to allocate task with relatively larger (smaller) cost to a position with relatively smaller (larger) weight. Figure \ref{fig:counter_example} shows a simple problem setting asking to assign a sequence of $80$ tasks to $4$ agents. There are only two task-types $H$ and $L$ in the task sequence, and $H<L$. The leftside shows the allocation scheme of the naive greedy strategy under this instance. The rightside shows the allocation scheme of the algorithm proposed in this paper, in which some L tasks are exchanged with some H tasks in lower positions, resulting a lower total cost due to the rearrangement inequality.

\begin{figure*} [h!]
	\begin{centering}
	\includegraphics[width=1.0 \textwidth]{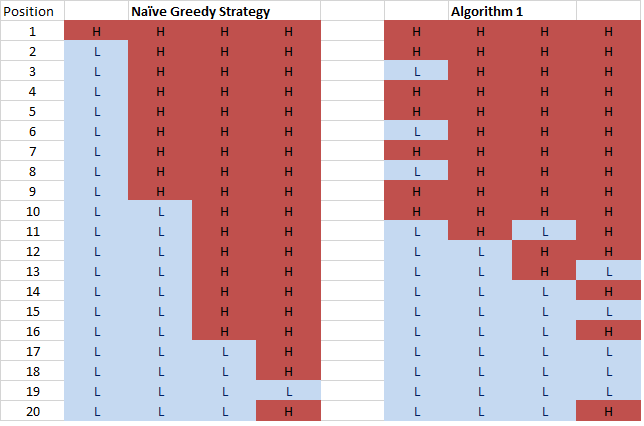}
	\caption{Diagram illustration of the allocation scheme of the naive greedy algorithm and Algorithm \ref{algo:greedy} under the instance HHLHH LHLHH LLLHH HHLLH HHLLH HLLHH HHLLL LHHLH HHLHL LHHLL LHHHL HHHHH LHHLH HLLHL LHLHL HLLLH}
	\label{fig:counter_example}
	\end{centering}
\end{figure*}

\section{Proofs}

\subsection{Proof of Lemma~\ref{lemma:allocation}} \label{proof:allocation}
\begin{proof}
It is obvious that we can get unique allocation scheme from a decision sequence. We prove the opposite direction by giving an algorithm (Algorithm \ref{algo:decoding}) that ``decodes" how the tasks are assigned by the input allocation scheme.
\begin{algorithm} [h!] \label{algo:decoding}
\KwIn{$\aaA{1} \dots \aaA{k}$}
\KwOut{$\x$}
set $p_r \leftarrow 1$ for $1\leq r \leq k$\;
$t \gets 1$\;
\While{$\exists r \in [k]$,  $p_r \leq |\aaA{r}|$}
{
	\uIf{$\exists r \in [k]$, $\aaA{r} = t$}
	{
		$x_t \gets r$; $p_r \gets p_r + 1$; $t \gets t + 1$ \;
	} \Else
	{\Return ``This is not a valid allocation scheme."\; \label{algo:line}}		
}
\Return $\x$ \;
\caption{Algorithm that decodes the decision sequence from the input allocation scheme}
\end{algorithm}

For contradiction, suppose Algorithm \ref{algo:decoding} fails to output a decision sequence (i.e. goes into Line \ref{algo:line}) at time $t$ when the input $\mathbf{A} = (\aaA{1}, \dots, \aaA{k} )$ satisfies property (1)(2)(3) in Definition \ref{def:allocation}. Since every task before $t$ has been decoded, we know $p_r>t$ for all $r \in [k]$. Due to the monotonicity of allocation schemes we know $t \not\in \aaA{r}$ for any $r \in [k]$, which violates the completeness of allocation schemes. 
\end{proof}

\subsection{Proof of Lemma~\ref{lemma:capacity}} \label{proof_capacity}
\begin{proof}
The capacity of agent decreases $1$ for each task assignment. Due to line $9 \sim 10$, Algorithm \ref{algo:greedy} never choose an agent $i$ that has the same capacity with another agent $j$ but has larger index of agent (i.e., $m_i=m_j$ and $i>j$), so the order of any agent pairs with respect to their capacities cannot be reversed.
\end{proof}

\subsection{Proof of Lemma~\ref{lemma:horizon}} \label{proof_low_greedy}
\begin{proof}
For a given look-ahead setup $(r,h)$, we will show that $Z^{(r,h,t)}_H + Z^{(r,h,t)}_L$ decreases at most $1$ for any single task assignment, which directly yields 
\begin{displaymath}
(Z^{(r,h,i)}_H + Z^{(r,h,t)}_L) - (Z^{(r,h,j)}_H + Z^{(r,h,j)}_L) < j-i \text{~~for any~} i<j
.\end{displaymath}
To prove $(Z^{(r,h,t)}_H + Z^{(r,h,t)}_L) - (Z^{(r,h,t+1)}_H + Z^{(r,h,t+1)}_L) \leq 1$ for any task $t$, let $p^*$ be the current position of agent $r-1$ at time $t$ (i.e. $p^*=m+1-\Q{r-1}{t}$), and let $r^*$ be any assignment decision of task $t$ made by Algorithm \ref{algo:greedy} (i.e. $r^* = x_t$). It is obvious that $Z_H+Z_L$ will not change if $r^* > h$ or $r^* < r-1$, since the task $t$ is assigned to an agent not covered by either H-zone or L-zone (i.e. $r^* \not\in S_H$ and $r^* \not\in S_L$). In cases of $r \leq r^* leq h$, the size of L-zone $Z_L$ doesn't change because the ``threshold position" $p^*$ doesn't change; and the size of H-zone $Z_H$ decrease $1$ due to the assignment of task $t$ to agent $r^* \in S_H$. 

In cases of $r^* = r-1$, the ``threshold position" $p^*$ decreases $1$, and the H-zone and L-zone will ``get" and ``lose", respectively, exactly one position for each agent $\gamma \in \{r\dots h\}$. In addition, the L-zone will lose one position for agent $r-1$ (i.e. the threshold position $p^*$). Due to Lemma \ref{lemma:capacity}, we have $m_{r-2} < m_{r-1}$ if Algorithm \ref{algo:greedy} chooses $r-1$, which means none of agents $\gamma < r-1$ has a position of $p^*$ at this time, so the L-zone will not further lose positions for them despite of the decreasing of the threshold position $p^*$. In total, the sum of $Z_H$ and $Z_L$ decrease $1$ in this case.
\end{proof}

\subsection{Proof of Lemma~\ref{lemma:low_greedy}} \label{proof_low_greedy}
\begin{proof}
For the decision sequence $\mathbf{x} = (x_1,\dots,x_n)$, let $x_1 = \gamma$. By definition $\gamma > 1$. Let $A = (\aaA{1}, \aaA{2}, \dots, \aaA{k})$ be the allocation scheme corresponding to $\mathbf{x}$.

Now consider $\aaA{1}$ and $\aaA{\gamma}$ in the allocation scheme. If $m_{1} = m_{\gamma}$, we can simply construct $\mathbf{x}^*$ by switching the tasks in the two agents $1$ and $\gamma$. Formally, this means we construct $A^* = (\aaA{\gamma}, \aaA{2} \dots \aaA{\gamma-1}, \aaA{1}, \aaA{\gamma+1} \dots \aaA{k})$, and by Lemma \ref{lemma:allocation} we can in turn construct $\mathbf{x}^*$ from $A^*$. So, in the following we assume $m_1 < m_{\gamma}$. Because task $1$ is assigned to agent $\gamma$, which is of capacity $\gamma$, we have
\begin{eqnarray*}
&&\aaA{1} = (\aaa{1}{1} \dots \aaa{1}{m_1-1}, \aaa{1}{m_1})^\T \\
&&\aaA{\gamma} = (\aaa{\gamma}{1} \dots \aaa{\gamma}{m_1-1}, \aaa{\gamma}{m_1}, \aaa{\gamma}{m_1+1} \dots \aaa{\gamma}{m_\gamma-1}, 1)^\T
.\end{eqnarray*}
Since by definition $\aaa{1}{m_1} \neq \aaa{\gamma}{m_1}$, there can be two cases:
 
\textbf{Case 1: } When $\aaa{1}{m_1} < \aaa{\gamma}{m_1}$. Intuitively this means the agent $1$ receives its first task ($\aaa{1}{m_1}$) before the capacity of agent $\gamma$ goes down to below $m_1$. In that case, we exchange the allocation target for task $1$ and $\aaa{1}{m_1}$ (i.e. assign task $1$ to agent $1$ and assign task $\aaa{1}{m_1}$ to agent $\gamma$). Due to Lemma \ref{lemma:exchange}, this is equivalent to a chain of exchanges, each either reduces the total cost or keeps it the same. Formally, assume $\aaa{\gamma}{s} < \aaa{1}{m_1} < \aaa{\gamma}{s+1}$ for some $m_1 < s \leq m_{\gamma}$ (such an $s$ always exists), we construct 
\begin{eqnarray*}
&&\aaA{*1} = (\aaa{1}{1} \dots \aaa{1}{m_1-1}, 1)^\T \\
&&\aaA{*\gamma} = (\aaa{\gamma}{1} \dots \aaa{\gamma}{m_1-1}, \aaa{\gamma}{m_1},\aaa{\gamma}{m_1+1} \dots \aaa{\gamma}{s-1}, \aaa{1}{m_1}, \aaa{\gamma}{s} \dots \aaa{\gamma}{m_\gamma-1})^\T
,\end{eqnarray*}
and $~~~~~~~~~A^* = (\aaA{*1}, \aaA{2} \dots \aaA{\gamma-1}, \aaA{*\gamma}, \aaA{\gamma+1} \dots \aaA{k})$. 

\noindent Recall that $A = (\aaA{1}, \aaA{2} \dots \aaA{\gamma-1}, \aaA{\gamma}, \aaA{\gamma+1} \dots \aaA{k})$, then we have
\begin{eqnarray*}
&& COST_{\mathbf{I}}(\mathbf{x}^*) - COST_{\mathbf{I}}(\mathbf{x}) \\
&=& \Big( g(m_1)L + g(s)T_{1,m_1} + \sum_{l=s+1}^{m_{\gamma}} g(l)T_{\gamma,l-1} \Big)
- \Big( g(m_1)T_{1,m_1} + \sum_{l=s}^{m_{\gamma}-1} g(l)T_{\gamma,l} + g(m_{\gamma})L \Big) \\
&=& L \Big( g(m_1)-g(m_{\gamma}) \Big) + T_{1,m_1} \Big( g(s)-g(m_1) \Big) + \sum_{l=s}^{m_\gamma-1} T_{\gamma,l} \Big( g(l+1) - g(l) \Big) \\
&\leq& L \Big( g(m_1)-g(m_{\gamma}) \Big) + L \Big( g(s)-g(m_1) \Big) + \sum_{l=s}^{m_\gamma-1} L \Big( g(l+1) - g(l) \Big) \\
&=& 0 
.\end{eqnarray*}

\textbf{Case 2: } When $\aaa{\gamma}{m_1} < \aaa{1}{m_1}$. This means the capacity of agent $\gamma$ goes down to below $m_1$ before the agent $1$ is ever assigned any task (thus still having a capacity of $m_{\gamma}$ by then). In this case we simply exchange all tasks of agents $1$ with the tasks assigned to agent $\gamma$ when the capacity of $\gamma$ is no more than $m_1$, yielding
\begin{eqnarray*}
\aaA{'1} &=& (\aaa{\gamma}{1} \dots \aaa{\gamma}{m_1}) \\
\aaA{'\gamma} &=& (\aaa{1}{1} \dots \aaa{1}{m_1}, \aaa{\gamma}{m_1+1} \dots \aaa{\gamma}{m_\gamma-1}, 1)
.\end{eqnarray*}
Clearly the exchange will not change the total cost, thus reduce the problem to Case 1. 
\qed \end{proof}

\subsection{Proof Sketch of Lemma~\ref{lemma:high_greedy}} \label{proof_high_greedy}
\begin{proofsketch}
The complete proof is rather long, so in this paper we will omit some repeated details when it is safe to do that, especially for rigorous proofs of the superiority of a rearranged sequence like in Lemma \ref{lemma:low_greedy}.

It is easy to check that Lemma \ref{lemma:high_greedy} holds if $m_r = 1$ for each agent $r \in [k]$ (any algorithm gives the same total cost in these cases). For general $(m_1 \dots m_k) \in \mathbf{N}^k$, assume for induction that Lemma \ref{lemma:high_greedy} holds for all the ``smaller" instances, that is, that Algorithm \ref{algo:greedy} minimizes the total cost of any instance $(m_1-\delta_1, m_2-\delta_2, \dots, m_k-\delta_k, T_{1+\sum \delta_r}, \dots, T_n)$. We will prove that Algorithm \ref{algo:greedy} will also minimize the cost of the instance  $(m_1, \dots, m_k, T_1,\dots, T_n)$. Specifically, for any decision sequence $\mathbf{x} = (x_1 \dots x_n)$, let $x_1 = \gamma$, we only need to prove in two cases. For convenience we will denote $H_i$ as the $i$-th H task in the sequence starting from task $1$. Similarly we denote $L_i$ as the $i$-th L task in the sequence.

\textbf{When $\gamma<\gamma^*$: } Without loss of generality we can assume that $x_2 \dots x_n$ follows Algorithm \ref{algo:greedy}, for otherwise we can simply turn to consider such a $\mathbf{x}'$, which guarantees to have lower costs than $\mathbf{x}$ due to the assumption of the induction. For any such instance $(m_1 \dots m_k, H, T_2 \dots T_n)$ and any such sequence $(\gamma, x_2 \dots x_n)$, we can have the following observations, which collectively characterize a ``overflowing" situation. 

First, we know $m_{\gamma} < m_{\gamma^*}$ because Algorithm \ref{algo:greedy} always breaks ties by returning the agent with smaller id. Second, since Algorithm \ref{algo:greedy} returns $\gamma^*$ for $T_1 = H$, by definition we know that there must exist $h$ such that there are at least $Z_L$ ``L"s in $T_1 \dots T_{Z_L+Z_H}$, where $Z_L = \sum_{i=1}^{h-1} \min \{m_i, m_{\gamma^*-1}\}$ and $Z_H = \sum_{i=\gamma^*}^{h} (m_i - m_{\gamma^*-1})$ are the sizes of H-zone and L-zone at this time, respectively. Third, we know $H_1=1$ is the only H task assigned by $\mathbf{x}$ to an agent with smaller id with $\gamma^*$, as shown by the following claim. 
\begin{claim}
For any $i > 1$, we have $x_i \geq \gamma^*$ if $T_i = H$.
\end{claim}
\begin{proof}
We know $x_{H_2} \geq \gamma^*$ because Algorithm \ref{algo:greedy} has chosen $\gamma^*$ for $H_1$ and allocating $H_1$ to agent $\gamma < \gamma^*$ only decreases the stopping threshold of Algorithm \ref{algo:greedy}. Furthermore, because we have assumed by induction that $x_2 \dots x_n$ is optimal for $T_2 \dots T_n$, by Lemma \ref{lemma:decision} we know that all the H tasks after $H_2$ will also be allocated to agents with id no smaller than $\gamma^*$. 
\end{proof}

In addition, the following claim shows that the sequence $\mathbf{x}$ (where $x_1= \gamma$ and $x_2 \dots x_n$ follows Algorithm \ref{algo:greedy}) will allocate at least one ``L" in the H-zone.  
\begin{claim} 
If allocating according to the sequence $x_1 \dots x_n$, there always exists an $i$ such that $1 \leq i \leq Z_L+Z_H$, $T_i = L$, $p_{i} > m_{\gamma^*-1}$, and $\gamma^* \leq x_i \leq h$.
\end{claim}

Combining all the above observations together, we can derive the following claim, which asserts that there must be some L overflowing from the L-zone to the H-zone if we follow the sequence $\mathbf{x}$ when $\gamma < \gamma^*$. Let $L^*$ denote the first such ``$L_i$". Observe that the decision sequence $\mathbf{x}$ has put $H_1$ (which is also $T_1$) in the L-zone (i.e. "below" position $m_r$) while put $L^*$ in the H-zone (i.e. above position $m_r$). Because $L^*$ is the first ``L" above $m_r$, we know all the tasks on top of $L^*$ are ``H".
\begin{claim} \label{lemma:overflow_L}
If allocating according to the sequence $x_1 \dots x_n$, there always exist $1 \leq H^* < L^* \leq n$ such that $p_{H^*} < p_{L^*}$, and that for any $H^* < i < L^*$, we have $T_i = L$ if $x_i = x_{H^*}$, and $T_i = H$ if $x_i = x_{L^*}$.
\end{claim}

Based on this claim we can construct a better sequence $\mathbf{x}^*$ by allocating $H_1$ to agent $x_{L^*}$ and allocating $L^*$ to agent $\gamma$. 
\begin{equation*}
\mathbf{x}^* = ( x_{L^*}, x_2 \dots x_{L^*-1}, \gamma, x_{L^*+1} \dots x_n )^\T
\end{equation*}
Compared with the one of $\mathbf{x}$, everything is the same except that a pair of H and L is exchanged in positions, which always lowers down the total cost due to Lemma \ref{lemma:exchange}). The rigorous proof of the benefit of this exchange is similar to the proof of Lemma \ref{lemma:low_greedy}. 

\textbf{When $\gamma>\gamma^*$: } Informally these are the cases when $\mathbf{x}$ put the first task in somewhere ``higher" (i.e., in a smaller position) than where it ``should have been". In the following we will prove that, for any instance $(m_1 \dots m_k, H, T_2 \dots T_n)$ and any sequence $(\gamma, x_2 \dots x_n)$ with $\gamma > \gamma^*$, there always exists a sequence $\mathbf{x}' = (x'_1 \dots x'_n)$ with $x'_1=\gamma' < \gamma$ such that $COST_I(\mathbf{x}') < COST_I(\mathbf{x})$. In other words, any sequence $\mathbf{x}$ with $x_1 > \gamma^*$ cannot be the optimal sequence.

Again, by induction we can assume that $x_2 \dots x_n$ follows Algorithm \ref{algo:greedy}. Similar with the first case, we assume for contradiction that $x_1 \dots x_n$ is optimal, which leads to a series of observations that collectively characterize a snapshot of the allocation. Then we will do some task exchanges in the allocation scheme to reduce the total cost without violating the monotonicity of the allocation scheme, thus forming a contradiction.

First, without loss of generality we know $m_{\gamma-1} < m_{\gamma}$, for otherwise if $m_{\gamma-1}=m_{\gamma}$ we will simply find the smallest agent id $r$ with $m_r = m_{\gamma}$, exchange the tasks allocated in agent $\gamma$ and $r$, and turn to consider the new allocation scheme.

Second, let $H_2$ denote the second H task in $T_1 \dots T_n$ (i.e. $T_i = L$ for any $1<i<H_2$), we know $H_2$ must not be put at the left side of $\gamma$ if $\mathbf{x}$ wants to be optimal. That is, we have $x_{H_2} \geq x_{H_1} = \gamma$. This is because $x_{H_2} < x_{H_1}$ will lead to $p_{H_2} < p_{H_1}$ (note that we just showed $m_{\gamma-1} < m_{\gamma}$), in which case we can turn to consider the sequence $x_{H_2}, x_2 \dots, x_{H_1}, \dots x_n$, which guarantees to have no more cost than $x_1 \dots x_n$. The complete proof is similar to the proof of Lemma \ref{lemma:low_greedy}. Note that by Lemma \ref{lemma:decision} all the subsequent H tasks will also be assign to agent $\gamma$ or at its right side.

Third, we know that $\mathbf{x}$ allocates at least one H task below (including) the position $m_{\gamma-1}$. Formally, we have 
\begin{claim}
There must exist an $i$ such that $T_i = H$ and $p_{i} \leq n_{\gamma-1}$.
\end{claim}

Let $H^*$ be the \emph{first} such H task (which basically ``overflow"s from the H-zone). We know that there is no interleaved H and L in the H-zone at least until $H^*$. Formally, we have
\begin{claim}
For any $1 \leq i < j \leq H^*$, if $x_i = x_j$ and $T_j = H$, then we must have $T_i = H$.  
\end{claim}

Finally, let $L^*$ be the first L task \emph{after} $H^*$ (such a $L^*$ must exist, for otherwise $H^*$ would be put somewhere lower). In the snapshot of the situation right after the $L^*$ is allocated, we can construct a better decision sequence $\mathbf{x}'$ by allocating the first task in agent $\gamma-1$ (rather than in $\gamma$ as the original sequence $\mathbf{x}$ does) and allocating $L^*$ at the position of $L^*-1$. It is easy to see that the re-arrangement virtually exchange a pair of H and L tasks, which will not increase the total cost due to Lemma \ref{lemma:exchange}. On the other hand, to show that such a re-arrangement will not violate the monotonicity of the allocation scheme, the key insight is to see that, by definition of Algorithm \ref{algo:greedy} there are no enough L tasks to ``catch up with" the allocation pace of H as long as H is not overflowing from the H-zone.
\end{proofsketch}

\section{A Faster Version of Algorithm \ref{algo:greedy}} \label{sec_faster_algo}

\begin{algorithm}[h!] \label{algo:threshold}
\caption{A faster version of the simulation-based algorithm}

\SetKw{Continue}{continue}
\SetKw{Goto}{goto}
\SetKwFunction{SA}{StreamingAllocation} 
\newcommand{\hrn}{\mathcal{H}_{r,h}}
\newcommand{\trd}{\mathcal{T}_{r,h}}
\newcommand{\cnt}{\mathcal{C}_{r,h}}

\KwIn{$ m_1 \dots m_k, T_1 \dots T_n$, where $m_1\leq m_2 \leq \dots \leq m_k$}
\KwOut{$x_1 \dots x_n$}
\BlankLine

let $\hrn$ be a $K \times K$ matrix denoting the look-ahead horizons\\
let $\trd$ be a $K \times K$ matrix denoting the stopping thresholds\\
let $\cnt$ be a $K \times K$ matrix denoting the counters of L tasks\\
\BlankLine

\For{$r = k$ \KwTo $2$}
{
	\For{$h = \gamma$ \KwTo $k$}
	{  
		$\trd \gets \sum_{i = 1}^{h-1} \min \{m_i, m_{r-1}\}$ \\ 
		$\hrn \gets \trd + \sum_{i = r}^h (m_i - m_{r-1})$  \\
		$\cnt \gets |\{i: i \in \{1, \dots, \hrn\} \text{~~and~~} T_i = L\}|$ \\  
	}
}
\BlankLine

set $\Q{r}{1} = m_r$ for each $1 \leq r \leq k$

\For{$t = 1$ \KwTo $n$}
{
	$x_t \gets$ \SA ($\Q{1}{t} \dots \Q{k}{t}, T_1 \dots T_n, t$)  

	$\Q{r}{t+1} \gets \Q{r}{t} -\mathds{1}(x_t = r)$ ~~for each $1 \leq r \leq k$
}

\Return $\mathbf{x}$  \\
\end{algorithm}

\begin{algorithm}[h!]
\caption{The routine in Algorithm \ref{algo:threshold}, which updates the data variables incrementally.}

\SetKw{Continue}{continue}
\SetKw{Goto}{goto}
\SetKwFunction{SA}{StreamingAllocation} 
\newcommand{\hrn}{\mathcal{H}_{r,h}}
\newcommand{\trd}{\mathcal{T}_{r,h}}
\newcommand{\cnt}{\mathcal{C}_{r,h}}

\KwSty{Function} \SA \\
\KwIn{$ m_1 \dots m_k, T_1 \dots T_n, t$, where $m_1\leq m_2 \leq \dots \leq m_k$}
\KwOut{the agent to which $T_t$ is assigned to}
\BlankLine

	\uIf{$T_t = L$}
	{
		$r^* \gets$ the smallest $r$ with $m_{r} > 0$ 
	} 
	\Else
	{
		\For{$r = k$ \KwTo $2$}
		{
			\For{$h = \gamma$ \KwTo $k$}
			{
				\uIf{$\cnt \geq \trd$} 
				{
					$r^* \gets r$ \\
					\Goto line \ref{line:settled} \\
				}
			}
		} 
		$r^* \gets 1$ \\
	}
	\BlankLine

	\label{line:settled}
	\For{$r = k$ \KwTo $2$} 
	{
		\For{$h = \gamma$ \KwTo $k$}
		{
			\uIf{$r^* > h$} 
			{
				$\hrn \gets \hrn +1$ \\
				$\trd \gets \trd$ \\
				$\cnt \gets \cnt - \mathds{1}(T_t = L) + \mathds{1}(T_{\hrn} = L)$
			}
			\uElseIf{$r \leq r^* \leq h$}
			{
				$\hrn \gets \hrn$ \\
				$\trd \gets \trd$ \\
				$\cnt \gets \cnt - \mathbb{I}(T_t = L)$
			}
			\uElseIf{$r^* = r-1$}
			{
				$\hrn \gets \hrn +1$ \\
				$\trd \gets \trd -(1+h-r)$ \\
				$\cnt \gets \cnt - \mathds{1}(T_t = L) + \mathds{1}(T_{\hrn} = L)$
			}
			\ElseIf{$r^* < r-1$}
			{
				$\hrn \gets \hrn$ \\
				$\trd \gets \trd -1$ \\
				$\cnt \gets \cnt - \mathds{1}(T_t = L)$
			}
		}
	} 
	\BlankLine

	\Return $r^*$ \\
\end{algorithm}